\documentclass[12pt,a4paper]{article}
\usepackage[cm]{fullpage}
\pdfoutput=1
\date{}

\usepackage[american]{babel}
\usepackage{url}
\usepackage{epsfig} 
\usepackage{times} 
\usepackage{amsmath}
\usepackage{amsfonts}
\usepackage{amssymb}  
\usepackage{tabularx}

\usepackage{xcolor}
\usepackage{transparent}
\usepackage{graphicx}
\graphicspath{{media/}}
\usepackage{tikz}
\usetikzlibrary{shapes,arrows}
\usepackage{epstopdf}
\usepackage{units}
\usepackage{pgf,pgfplots}
\usepackage{algorithm}
\usepackage{algpseudocode}
\usepackage{enumitem} 

\usepackage{pstool}

\graphicspath{{media/}}
\newcommand{\executeiffilenewer}[3]{%
 \ifnum\pdfstrcmp{\pdffilemoddate{#1}}%
 {\pdffilemoddate{#2}}>0%
 {\immediate\write18{#3}}\fi%
}

\usepackage{arydshln}
\usepackage{cleveref}

\newtheorem{theorem}{Theorem}[section]
\newtheorem{lemma}[theorem]{Lemma}
\newtheorem{proposition}[theorem]{Proposition}

\newenvironment{proof}[1][Proof]{\begin{trivlist}
\item[\hskip \labelsep {\bfseries #1}]}{\end{trivlist}}

\newcommand{\qed}{\nobreak \ifvmode \relax \else
      \ifdim\lastskip<1.5em \hskip-\lastskip
      \hskip1.5em plus0em minus0.5em \fi \nobreak
      \vrule height0.5em width0.5em depth0.25em\fi}
      
\newcommand{\TT}{\ensuremath{\mathsf{\tiny{T}}}}
\newcommand{\T}{^{\TT}}

\newcommand{\diff}[1][]{\mathrm{d}#1}
\newcommand{\dt}{\diff t }

\author{Michael Muehlebach and Raffaello D'Andrea
\thanks{Michael Muehlebach and Raffaello D'Andrea are with the Institute for Dynamic Systems and Control, ETH Zurich. The contact author is Michael Muehlebach, {\tt\small michaemu@ethz.ch}.
This work was supported by ETH-Grant ETH-48 15-1.}%
}
\title{\LARGE \bf
On the Approximation of Constrained Linear Quadratic Regulator Problems and their Application to Model Predictive Control - Supplementary Notes}

\begin{document}

\maketitle

\begin{abstract}
By parametrizing input and state trajectories with basis functions different approximations to the constrained linear quadratic regulator problem are obtained. These notes present and discuss technical results that are intended to supplement a corresponding journal article. The results can be applied in a model predictive control context.
\end{abstract}

\section{Outline}
%
Sec.~\ref{Sec:Not} summarizes the notation. The resulting approximations are given by convex finite-dimensional optimization problems that have a quadratic cost, linear equality constraints, and linear semi-infinite inequality constraints. Sec.~\ref{Sec:SemiinfiniteConst} suggests several strategies for dealing with these semi-infinite constraints. One of these strategies is discussed in more details in Sec.~\ref{Sec:activeSet}, as it is found to be computationally efficient in a model predictive control context. The remaining sections discuss several technical results. Sec.~\ref{App:PropB} and Sec.~\ref{App:PropC} are related to an upper, respectively lower bound approximation of the cost of the underlying constrained linear quadratic regulator problem. Sec.~\ref{App:SetConv} presents a result that is related to the convergence of the different approximations. Sec.~\ref{App:Stab} shows recursive feasibility and closed-loop stability of the resulting model predictive control algorithm. Sec.~\ref{App:Decay} and Sec.~\ref{App:BasisFunProp} highlight important properties of the basis functions parametrizing input and state trajectories.
\section{Notation}\label{Sec:Not}
We are concerned with the approximation of the following optimal control problem
\begin{align}
J_\infty := &\min \frac{1}{2} \left( ||x||^2 + ||u||^2 \right) + \psi(x_T) \label{eq:OrgProb}
\\
&\text{s.t.}~~\dot{x}=A x+ B u, x(0)=x_0, x(T)=x_T,\nonumber \\
&\qquad C_\text{x} x+C_\text{u} u \leq b, x_T \in \mathcal{X},\nonumber \\
&\qquad x \in L_n^2, u \in L_m^2, \nonumber
\end{align}
where the space of square integrable functions mapping from the interval $I:=(0,T)$ to $\mathbb{R}^q$ is denoted by $L_q^2$, where $q$ is a positive integer; and the $L_q^2$-norm is defined as the map $L_q^2 \rightarrow \mathbb{R}$,
\begin{equation}
x \rightarrow ||x||, \quad ||x||^2 := \int_{I} x\T x~\dt,
\end{equation}
with $\dt$ the Lebesgue measure. The function $\psi: \mathbb{R}^n \rightarrow \mathbb{R}$ is assumed to be positive definite and strongly convex, $A\in \mathbb{R}^{n\times n}$, $B\in \mathbb{R}^{n \times m}$, $C_\text{x} \in \mathbb{R}^{n_c \times n}$, $C_\text{u} \in \mathbb{R}^{n_c \times m}$, and $b\in \mathbb{R}^{n_c}$ are constant, and the set $\mathcal{X}$ is closed and convex. The dynamics as well as the stage constraints are assumed to be fulfilled almost everywhere. Thus we simply write
\begin{equation}
f=g, \quad f\leq g,
\end{equation}
when we mean $f(t)=g(t)$, respectively $f(t)\leq g(t)$ for all $t\in I$ almost everywhere, with $f,g \in L^2_q$.
The weak derivative of $x$ is denoted by $\dot{x}$.\footnote{The equations of motion imply that $\dot{x} \in L^2_n$, which can be used to conclude that $x$ has a unique absolutely continuous representative defined on the closure of $I$ (a classical solution of the equations of motion). With $x(t)$ we refer to the value this unique absolutely continuous representative takes at time $t\in [0,T]$.} To simplify notation we abbreviate the domain of the objective function by 
\begin{equation}
X:=L^2_n \times L^2_m \times \mathbb{R}^n.
\end{equation}

We assume throughout the article that the constraints in \eqref{eq:OrgProb} are nonempty, i.e. there exists trajectories $x$ and $u$, fulfilling the dynamics, the initial condition, the constraints, and thus achieve a finite cost. 

The main motivation for studying problem \eqref{eq:OrgProb} comes from the fact that \eqref{eq:OrgProb} often serves as a starting point for model predictive control (MPC).

The problem \eqref{eq:OrgProb} can be written in compact form as
\begin{align}
J_\infty = &\min ||x||^2 + ||u||^2 + \psi(x_T) \nonumber\\
&\text{s.t.}~(x,u,x_T)\in \mathcal{C} \cap \mathcal{D}, \label{eq:RefOrgProb}
\end{align}
where the set $\mathcal{C}$ contains all trajectories $x \in L^2_n, u \in L^2_m$ that fulfill the constraints and $\mathcal{D}$ the trajectories that fulfill the dynamics and the initial condition, that is,
\begin{align}
\mathcal{C}&:=\{ (x,u,x_T) \in X \;|\; C_\text{x} x + C_\text{u} u \leq b, x_T \in \mathcal{X}\}, \label{eq:defMX}\\
\mathcal{D}&:=\{(x,u,x_T)\in X | \dot{x}=Ax + Bu, x(0)=x_0, x(T)=x_T\}. \label{eq:defD}
\end{align}
It was shown that \eqref{eq:RefOrgProb} can be approximated by two auxiliary optimization problems, where input and state trajectories are parameterized with basis functions,
\begin{equation*}
\tilde{x}(t)=(I_n \otimes \tau^s(t))\T \eta_x, \quad \tilde{u}(t)=(I_m \otimes \tau^s(t))\T \eta_u.
\end{equation*}
The parameter vectors are denoted by $\eta_x \in \mathbb{R}^{ns}$ and $\eta_u \in \mathbb{R}^{ms}$, whereas $\tau^s(t):=(\tau_1(t), \tau_2(t), \dots, \tau_s(t))\in \mathbb{R}^s$ contains the first $s$ basis functions, $\otimes$ denotes the Kronecker product, and $I_q \in \mathbb{R}^{q\times q}$ refers to the identity matrix for any integer $q>0$. The superscript $s$ refers to the number of basis functions used for the approximation. For ease of notation the superscript $s$ will be dropped whenever the number of basis functions is clear from context. The basis functions are required to satisfy the following assumptions:
\begin{itemize}
\item[A1)] The basis functions $\tau_i \in L_1^2$, $i=1,2,\dots,s$ are linearly independent and orthonormal with respect to the standard $L_1^2$-scalar product.
\item[A2)] The basis functions fulfill $\dot{\tau}^s(t) = M_s \tau^s(t)$ for all $t\in I$, for some $M_s \in \mathbb{R}^{s \times s}$.
\end{itemize}
The finite-dimensional subspace spanned by the first $s$ basis functions will be denoted by $X^s$,
\begin{equation}
X^s:=\{ (x,u,x_T)\in X \; | \; \eta_x\in \mathbb{R}^{ns}, \eta_u\in \mathbb{R}^{ms},
x=(I_n \otimes \tau^s)\T \eta_x, u=(I_m\otimes \tau^s)\T \eta_u\}.
\end{equation}
We can think of an element in $X^s$ not only as an element in $X$ (i.e. a tuple of a finite-dimensional vector and two square integrable functions), but also as a finite-dimensional vector given by the corresponding parameter vectors $\eta_x$ and $\eta_u$. To make this distinction explicit, we introduce the map $\pi^{qs}: L_q^2 \rightarrow \mathbb{R}^{qs}$, defined as
\begin{equation}
x \rightarrow \int_{I} (I_q \otimes \tau^s) x \dt, \label{eq:proj}
\end{equation}
which maps an arbitrary element $x\in L^2_n$ to its first $s$ basis function coefficients. Similarly, we define $\pi^s: X \rightarrow \mathbb{R}^{ns} \times \mathbb{R}^{ms} \times \mathbb{R}^n$ as
\begin{equation}
(x,u,x_T) \rightarrow (\pi^{ns}(x), \pi^{ms}(u),x_T).
\end{equation}
As a consequence, we write $\pi^s(x)$ for describing the finite dimensional representation of $x \in X^s$, that is, its representation in terms of the parameter vectors $\eta_x$ and $\eta_u$. The adjoint map $(\pi^{qs})^*: \mathbb{R}^{qs} \rightarrow L_q^2$ is given by
\begin{equation}
\eta \rightarrow (I_q \otimes \tau)\T \eta, \label{eq:incl}
\end{equation}
and is used to obtain the trajectory corresponding to the vector $\eta \in \mathbb{R}^{qs}$, containing the first $s$ basis function coefficients. Similarly, we define $(\pi^s)^*: \mathbb{R}^{ns} \times \mathbb{R}^{ms} \times \mathbb{R}^n \rightarrow X$ as
\begin{equation}
(\eta_x, \eta_u, x_T) \rightarrow  ((\pi^{ns})^*(\eta_x), (\pi^{ms})^*(\eta_u), x_T).
\end{equation}
The composition $(\pi^{s})^* \pi^{s}: X \rightarrow X$ yields the projection of an element $x\in X$ onto the subspace $X^s \subset X$.

The sets $\mathcal{C}$ and $\mathcal{D}$ are approximated in two different ways,
\begin{align}
\mathcal{C}^s_\text{U}&:=\mathcal{C} \cap X^s,\\
\mathcal{C}^s_\text{L}&:=\{ (x, u, x_T) \in X \; | \; \int_{I} \delta \tilde{p}\T (-C_\text{x} {x} - C_\text{u} {u} + b) \dt \geq 0  \nonumber\\
&\hspace{4cm}\forall \delta \tilde{p}=(I_{n_c} \otimes \tau)\T \delta \eta_p: \delta \tilde{p}\geq 0, \delta \eta_p\in \mathbb{R}^{n_c s}; x_T \in \mathcal{X} \}, \label{eq:deftX} \\
\mathcal{D}^s_\text{U}&:=\mathcal{D} \cap X^s,\\
\mathcal{D}^s_\text{L}&:=\{ (x,u,x_T)\in X \; | \;  \delta \tilde{p}=(I_n\otimes \tau)\T \delta \eta_p, \nonumber \\
&\!\!-\int_{I} \delta \dot{\tilde{p}}\T x \dt -\int_{I} \delta \tilde{p}\T (Ax + Bu)\dt 
- \delta \tilde{p}(0)\T x_0 + \delta \tilde{p}(T)\T x_T = 0, \forall \delta \eta_p \in \mathbb{R}^{ns}\}, \label{eq:deftD}
\end{align}
resulting in the following two approximations to \eqref{eq:RefOrgProb}
\begin{align}
J_s:=\inf ||\tilde{x}||^2 + ||\tilde{u}||^2 + \psi(x_T) \nonumber\\
\text{s.t.}~(\tilde{x},\tilde{u},x_T) \in \mathcal{C}^s_\text{U}\cap \mathcal{D}^s_\text{U}, \label{eq:approx1}
\end{align}
and
\begin{align}
\tilde{J}_s:=\min ||\tilde{x}||^2 + ||\tilde{u}||^2 + \psi(x_T)\nonumber\\
\text{s.t.}~(\tilde{x},\tilde{u},x_T) \in \mathcal{C}^s_\text{L}\cap \mathcal{D}^s_\text{L}. \label{eq:approx2}
\end{align}
By definition of the constraints $\mathcal{C}^s_\text{U}$ and $\mathcal{D}^s_\text{U}$, the minimizer of \eqref{eq:approx1} (for $s\geq s_0$) is required to be an element of $X^s$. Consequently, the problem \eqref{eq:approx1} is equivalent to
\begin{align}
J_s = &\inf |\eta_x|^2 + |\eta_u|^2 + \psi(x_T) \nonumber\\
& \text{s.t.}~(\eta_x, \eta_u, x_T) \in \pi^s(\mathcal{C}^s_\text{U}) \cap \pi^s(\mathcal{D}^s_\text{U}), \label{eq:approx3}
\end{align}
which corresponds to a convex finite-dimensional optimization problem.
Likewise, it is found that the convex finite-dimensional problem 
\begin{align}
\tilde{J}_s = &\min |\eta_x|^2 + |\eta_u|^2 + \psi(x_T) \nonumber \\
&\text{s.t.}~(\eta_x,\eta_u,x_T) \in \pi^s(\mathcal{C}^s_\text{L}) \cap \pi^s(\mathcal{D}^s_\text{L}),  \label{eq:approx4}
\end{align}
is equivalent to \eqref{eq:approx2} in the sense that its (unique) minimizer $(\eta_x,\eta_u,\bar{x}_T)$ is related to the minimizer $(x,u,x_T) \in X$ of \eqref{eq:approx2} by $x=(\pi^{ns})^*(\eta_x)$, $u=(\pi^{ms})^*(\eta_u)$, $\bar{x}_T=x_T$, and achieves the same cost.

The proposed approximations can be applied in the context of MPC. We will show that by repeatedly solving the infinite-horizon optimal control problem \eqref{eq:approx3} (with $I=(0,\infty)$, $\psi=0$, $x_T = \lim_{t\rightarrow \infty} x(t)=0$), recursive feasibility and closed-loop stability are inherent to the resulting MPC algorithm. For the sake of completeness, \eqref{eq:approx3} is written out as
\begin{align}
&\min~\eta_x\T (Q \otimes I_s) \eta_x + \eta_u\T (R \otimes I_s) \eta_u \label{eq:mpc1}\\
&~\text{s.t.} ~(I_n \otimes M_s\T - A \otimes I_s) \eta_x - (B \otimes I_s) \eta_u = 0,\nonumber \\
&~~~ \quad (I_n \otimes \tau(0))\T \eta_x = x_0, \nonumber\\
&~~~ \quad (C_\text{x} \otimes \tau(t))\T \eta_x + (C_\text{u} \otimes \tau(t))\T \eta_u \leq b, \forall t \in [0,\infty), \nonumber
\end{align}
where the input and state costs are weighted with the matrices $Q>0$ and $R>0$, which is common in MPC.

In the following an efficient numerical solution algorithm for \eqref{eq:mpc1} will be discussed.

\section{Implementation of the semi-infinite constraint}\label{Sec:SemiinfiniteConst}
The optimization \eqref{eq:mpc1} is a convex finite-dimensional optimization problem. However, it is not a quadratic program, as it includes the semi-infinite constraint
\begin{equation}
(C_\text{x} \otimes \tau(t))\T \eta_x + (C_\text{u} \otimes \tau(t))\T \eta_u \leq b, \forall t\in [0,\infty). \label{eq:const}
\end{equation}
As a result, \eqref{eq:mpc1} cannot be solved by a standard quadratic programming solver. Three different approaches to deal with the semi-infinite constraint are immediate:
\begin{itemize}
\item[1)] global polyhedral approximation
\item[2)] sum-of-squares approximation
\item[3)] local polyhedral approximation (active-set approach)
\end{itemize}
The first is based on a fixed polyhedral approximation, leading to a quadratic program. The second is based on exploiting polynomial basis functions for obtaining a characterization using linear matrix inequalities, whereas the third is based on an iterative constraint sampling scheme, resulting in a local polyhedral approximation. In the following subsections we will focus on approach 1 and 2. We will describe an efficient solution algorithm based on approach 3 in detail in Sec.~\ref{Sec:activeSet}.

It turns out that it is enough to check the constraint \eqref{eq:const} over a compact time interval, instead of the unbounded interval $t\in [0,\infty)$. This is because the basis functions are assumed to be linearly independent and exponentially decaying according to Assumption A1 and A2. A formal proof of this claim can be found in Sec.~\ref{App:Decay}. The compact time interval for which the constraint \eqref{eq:const} has to be checked is denoted by $[0,T_\text{c}]$ and depends on the choice of basis functions, on the order $s$, and in some cases also on the bound $b$ (see Sec.~\ref{App:Decay}).


\subsection{Global polyhedral approximation}
In order to construct a global polyhedral approximation of the set $\mathcal{C}^s_\text{U}$ we assume that an upper bound on the achievable cost $J_s$, denoted by $\bar{J}_s$, is available. In order to simplify the discussion, we assume further that for now $C_\text{x}$ and $C_\text{u}$ are row vectors, and that $b$ is a scalar. The proposed scheme can be readily extended to the case where $C_\text{x}$ and $C_\text{u}$ are matrices, and $b$ is a vector. The approximation is based on constraint sampling, where we tighten the constraint slightly to
\begin{equation}
C_\text{x} \tilde{x}(t_i) + C_\text{u} \tilde{u}(t_i) \leq (1-\epsilon) b, \label{eq:constrTight}
\end{equation}
with $\epsilon>0$, and where $t_i$, denotes the constraint sampling instances, $i=1,2,\dots, n_i$, which are to be determined. The algorithm is based on the following two steps: \begin{itemize}
\item[1)] Compute
\begin{align*}
h(t):= &\max ~C_\text{x} \tilde{x}(t) + C_\text{u} \tilde{u}(t) - b\\
&~\text{s.t.}~C_\text{x} \tilde{x}(t_i) + C_\text{u} \tilde{u}(t_i) \leq (1-\epsilon) b, ~i=1,\dots, n_i,\\
&\quad (\tilde{x},\tilde{u}, \lim_{t_\text{e}\rightarrow \infty} \tilde{x}(t_\text{e})) \in \mathcal{D}^s_\text{U},\\
&\quad ||\tilde{x}||^2 + ||\tilde{u}||^2 \leq \bar{J}_s,
\end{align*}
for all times $t\in I_s$, where $I_s$ contains a finite number of sampling instances (to be made precise below).
\item[2)] Find the local peaks of $h(t)$, denoted by $t_i^*$. Add each $t_i^*$ to the constraint sampling points  if $h(t_i^*)> -b \epsilon/2$. Repeat the procedure until $h(t)\leq -b\epsilon/2$ for all $t\in I_s$.
\end{itemize}
Note that the function $h(t)$ is again only evaluated at the discrete time points $t\in I_s$. The index set $I_s$ has to be chosen such that $h(t)\leq -b\epsilon/2$ for all $t\in I_s$ implies that $h(t) \leq 0$ for all $t\in [0,\infty)$. As remarked earlier, due to the fact that the basis functions are exponentially decaying and linearly independent it is enough to check $h(t)\leq 0$ for all $t\in [0,T_\text{c}]$, for a fixed time $T_\text{c}$, as $h(t)\leq 0$ for all $t\in (T_\text{c},\infty)$ will be fulfilled automatically. Moreover, a Lipschitz constant of 
\begin{equation}
C_\text{x}  \tilde{x}(t) + C_\text{u} \tilde{u}(t) - b \label{eq:constrTmp}
\end{equation}
can be found by using an upper bound on its time-derivative, that is, for example,
\begin{equation}
|\tau(t)\T M\T (C_\text{x} \eta_x + C_\text{u} \eta_u)|\leq |\tau(t)| |M\T (C_\text{x} \eta_x + C_\text{u} \eta_u)|,
\end{equation}
where the first term can be bounded for all $t\in [0,\infty)$ due to the fact that the basis functions are exponentially decaying and the second term can be bounded using the fact that the cost $J_s$ is below $\bar{J}_s$. 
We therefore choose the index set $I_s$ as
\begin{equation}
I_s=\{ t_k < T_\text{c} \;|\; t_k=k \frac{2L}{b\epsilon}, k=0,1,2,\dots\},
\end{equation}
where $L$ denotes a Lipschitz constant of \eqref{eq:constrTmp}.

It is important to note that the optimization in step 1 imposes the dynamics and the upper bound $\bar{J}_s$ on the cost. Both constraints tend to reduce the number of constraint sampling points $t_i$ greatly. The initial condition $x_0$ enters the optimization as an optimization variable. The optimization in step 1 represents a quadratically constrained linear program for each time instant $t\in I_s$, and as such, it can be solved using standard software packages. The whole procedure for determining the constraint sampling points is done offline. Once these time instances are found, the optimization problem that is solved online reduces to a quadratic program. The number of constraint sampling points $t_i$ is upper bounded by the cardinality of the index set $I_s$, and thus guaranteed to be finite. Due to the fact that the above procedure is greedy, it will not necessarily lead to the smallest number of constraint sampling points.


\subsection{Sum-of-squares approximation} In case exponentially decaying polynomials are used as basis functions, sum-of-squares techniques can be applied. In particular, it is shown in \cite{Nesterov2000} that the set
\begin{equation}
\{ \eta \in \mathbb{R}^{s} \;|\; \eta\T (1,t,\dots, t^{s-1}) \geq 0, \forall t\in[0,\infty) \}
\end{equation}
can be expressed using matrix inequalities that are linear in the coefficients $\eta$. In the case of exponentially decaying polynomials it is therefore enough to approximate the exponential decay by a polynomial upper bound (for example by appropriately truncating a Taylor series expansion at $0$), in order to approximate the constraint \eqref{eq:const} in a slightly conservative manner. As a result, by applying the results from \cite{Nesterov2000}, the optimization problem \eqref{eq:mpc1} is approximated by a semidefinite program that can be solved using standard optimization routines.
\section{A dedicated active-set approach}\label{Sec:activeSet}
In the following section we present an efficient optimization routine for solving \eqref{eq:approx3} (and likewise \eqref{eq:mpc1}). The method is an extension of traditional active set methods and generalizes to optimization problems with a linear quadratic cost function, linear equality constraints, and linear semi-infinite inequality constraints, i.e.
\begin{align}
\hat{J}(I_c):=&\min z\T H z \label{eq:optiprob}\\
&\text{s.t.}~A_\text{eq} z = b_\text{eq}, \\
&\quad~l_\text{b} \leq (I_{n_c} \otimes \tau(t))\T C_\text{z} z \leq l_\text{u},\quad \forall t\in I_c, \label{eq:ineq}
\end{align}
where $I_c$ is any subset of the non-negative real line. In case of \eqref{eq:mpc1}, the interval $I_c$ is taken to be $[0,\infty)$. Note that an optimization problem, whose objective function has a linear part, can be brought to the form \eqref{eq:optiprob} by completing the squares. It is assumed that the optimization problem \eqref{eq:optiprob} is feasible, that $l_\text{b} < 0$ and $l_\text{u} > 0$, and that the Hessian $H$ is positive definite, which guarantees existence and uniqueness of the corresponding minimizer\footnote{This is due to the fact that the constraints describe a closed convex set and due to the strong convexity of the objective function.}.

The method is based on the observation that if the set $I_c$ consists merely of a collection of time instants (constraint sampling instances) $t_i$, \eqref{eq:optiprob} reduces to a quadratic program that can be solved efficiently. Moreover, due to the fact that the basis functions fulfill Assumptions A1 and A2, a trajectory parametrized with the basis functions has a finite number of maxima and minima, as is shown in Sec.~\ref{App:BasisFunProp}. 
Consequently, \eqref{eq:optiprob} has only a finite number of active constraints. The collection of the time instants corresponding to these active constraints will be denoted by $I_c^*$. If this finite collection of constraint sampling instants is known ahead of time, one could simply solve \eqref{eq:optiprob} with respect to $I_c^*$ instead of $I_c$, resulting in $\hat{J}(I_c^*)=\hat{J}(I_c)$. In addition, for any subset $I_c^k$ of $I_c$ it holds that $\hat{J}(I_c)\geq \hat{J}(I_c^k)$, and likewise, if $I_c^k$ is a subset of $I_c^{k+1}$ we have that $\hat{J}(I_c^{k+1})\geq \hat{J}(I_c^k)$. Hence, a monotonically increasing sequence $J(I_c^k)$, bounded above by $\hat{J}(I_c)$ can be constructed using any sequence of sets $I_c^k$ that fulfill $I_c^k\subset I_c^{k+1} \subset \dots \subset I_c$. In particular, such sets are obtained by starting with an arbitrary initial guess $I_c^0$ containing a finite number of constraint sampling points (or even the empty set), and by adding at least one constraint violation point at each iteration. Moreover, at each iteration the inactive constraints contained in the set $I_c^k$ can be removed, as this will not alter the optimizer nor the optimal value $\hat{J}(I_c^k)$. In that way, the number of constraint sampling instances contained in $I_c^k$ remains finite. This motivates Alg.~\ref{Alg:iter1}, which solves \eqref{eq:optiprob} up to a given tolerance, by constructing an approximation to the set of active constraints $I_c^*$.

\begin{algorithm}[]
\caption{Iterative constraint sampling}
\renewcommand{\algorithmicrequire}{\textbf{Initialize:}}
\begin{algorithmic}[1]
\Require initial guess for the constraint sampling points: $I_c^0=\{t_1, t_2, \dots, t_N\}$;
maximum number of iterations: MAXITER;
constraint satisfaction tolerance: $\epsilon$;
\State $k=0$
\For{$k<$MAXITER}
\State solve \eqref{eq:optiprob} for $I_c^k$ $\rightarrow$ $z^k$, $\hat{J}(I_c^k)$
\If{infeasible}
\State abort
\ElsIf{$z^k$ fulfills \eqref{eq:ineq} for all $t\in I$ (with tol. $\epsilon$)}
\State algorithm converged
\State return $z^k$
\EndIf
\State find at least one constraint violation instant $\rightarrow$ $t_c$
\State remove inactive time instants in $I_c^k$
\State $I_c^{k+1}=I_c^k\cup \{t_c\}$,~$k=k+1$
\EndFor
\end{algorithmic}
\label{Alg:iter1}
\end{algorithm}

\begin{proposition}
Alg.~\ref{Alg:iter1} converges, that is, $\lim_{k\rightarrow \infty} \hat{J}(I_c^k)=\hat{J}(I_c^*)=\hat{J}(I_c)$. In order to achieve constraint violations smaller than $\epsilon$ at most
\begin{equation}
\frac{4 c_\tau (\hat{J}(I_c)-\hat{J}(I_c^0))}{\sigma \epsilon^2} \label{eq:bounditer}
\end{equation}
steps are required, where $c_\tau$ is defined as
\begin{equation}
c_\tau:=\sup_{t\in I_c} |\tau(t)|^2,
\end{equation} 
and $\sigma$ denotes the smallest eigenvalue of the Hessian $H$.
\end{proposition}
\begin{proof}
From the above arguments it can be concluded that $\hat{J}(I_c^k)$ is monotonically increasing whenever $\hat{J}(I_c^k)<\hat{J}(I_c)$ and upper bounded by $\hat{J}(I_c)$. Therefore the sequence $\hat{J}(I_c^k)$ converges.
It remains to show that $\hat{J}(I_c^k)$ converges to $\hat{J}(I_c)$.
The strong convexity of the objective function can be used to establish
\begin{equation}
|z^{k+1}-z^k|^2 \leq 4 \sigma^{-1} (\hat{J}(I_c^{k+1})-\hat{J}(I_c^k)),
\end{equation}
where $z^{k+1}$ and $z^k$ are the minimizer corresponding to $\hat{J}(I_c^k)$ and $\hat{J}(I_c^{k+1})$. As a result we can conclude that $z^k$ converges, and that $\lim_{k\rightarrow \infty} z^k$ is well-defined and satisfies the constraint \eqref{eq:ineq} (by a contradiction argument). It is therefore a feasible candidate for \eqref{eq:optiprob}, implying that $\lim_{k\rightarrow \infty} \hat{J}(I_c^{k}) \geq \hat{J}(I_c)$, which, combined with $\hat{J}(I_c^{k})\leq \hat{J}(I_c)$ for all integers $k$, leads to $\lim_{k\rightarrow\infty} \hat{J}(I_c^k) = \hat{J}(I_c)$. 

It remains to show that \eqref{eq:bounditer} is fulfilled. To that extent, let $\epsilon>0$ denote the smallest constraint violation that occurs within the first $N$ steps. For all $k\leq N-1$, there exists the constraint violation point $t_j^k$, which will be added to $I_c^k$, and therefore
\begin{equation}
\epsilon \leq |(e(t_j^k) \otimes \tau(t_j^k))\T (z^k-z^{k+1})|,
\end{equation}
where $e(t_j^k)$ is a canonical unit vector. Combining the Cauchy-Schwarz inequality and the above bound on $|z^{k+1}-z^k|$ results in
\begin{equation}
\epsilon^2 \leq 4 c_\tau \sigma^{-1} (\hat{J}(I_c^{k+1})-\hat{J}(I_c^k)).
\end{equation}
By summing over the first $N$ steps we arrive at
\begin{equation}
N\epsilon^2 \leq 4 \sigma^{-1} c_\tau (\hat{J}(I_c^N)-\hat{J}(I_c^0)) \leq 4 \sigma^{-1} c_\tau (\hat{J}(I_c^*)-\hat{J}(I_c^0)),
\end{equation}
since the sequence $\hat{J}(I_c^k)$ is strictly increasing and bounded above by $\hat{J}(I_c^*)$. Dividing by $\epsilon^2$ on both sides concludes the proof. \hfill \qed
\end{proof}

\subsection{Implementation details}
Alg.~\ref{Alg:iter1} can be naturally embedded in an active-set method. An introduction to active-set methods for solving quadratic programs can be found for example in \cite[Ch.~10]{Fletcher}. Starting from an initial guess of the active constraint sampling instants, which is denoted by $I_c^0$, the quadratic program with optimal cost $\hat{J}(I_c^0)$ is solved: This is done by initially assuming that all constraints in the set $I_c^0$ are active. The resulting optimization problem reduces to an equality constrained quadratic program, whose solution can be calculated by solving a linear system of equations. The Lagrange multipliers corresponding to \eqref{eq:ineq}, which are denoted by $\mu(t)$, $t\in I_c^0$, are evaluated subsequently. If all the constraints are indeed active, then the optimizer to the quadratic program with cost $\hat{J}(I_c^0)$ has been found. If, however, not all the constraints are found to be active, that is, if there are some Lagrange multipliers that are zero, the standard active-set procedure, see \cite[Ch.~10]{Fletcher} is used to find the subset of active constraints $I_a \subset I_c^0$. Provided that the active set $I_a \subset I_c^0$ and the optimizer to the quadratic program with cost $\hat{J}(I_c^0)$ has been found, the constraint \eqref{eq:ineq} is then checked for all $t\in I_c$. If no constraint violations occur, the solution to \eqref{eq:optiprob} has been found. If constraint violations occur, the time instant $t_c$ for which a violation occurs, is added to the set of active constraints $I_a$ resulting in $I_c^1=I_a \cup \{t_c\}$. The above procedure is then repeated until convergence. 

Each iteration requires solving equality constrained quadratic programs of the type 
\begin{align}
&\min z\T H z \label{eq:optequal}\\
&\text{s.t.}~A_\text{eq} z = b_\text{eq}, \quad (c(t) \otimes \tau(t))\T C_\text{z} z = l_\text{a}(t),\quad \forall t \in I_a,
\end{align}
where $c(t) \in \mathbb{R}^{n_c}$, $l_\text{a}(t)\in \mathbb{R}$, 
$t\in I_a$, and $I_a \subset I_c^k$ describe the active constraints corresponding to \eqref{eq:ineq}. Due to the fact that very few constraints are expected to be active, we use a range space approach, \cite[p.~238]{Fletcher}. To that extent, the equality constraint is eliminated and the optimizer $z^*$ corresponding to \eqref{eq:optequal} is rewritten as
\begin{equation}
z^*=\hat{b}+\hat{H} C_\text{z}\T \sum_{t\in I_a} (c(t)\otimes \tau(t))\T \mu(t),\label{eq:zstar}
\end{equation}
where the dual variable $\mu(t)\in \mathbb{R}$, defined for $t\in I_c^k$, satisfies
\begin{equation}
(c(t_j) \otimes \tau(t_j))\T C_\text{z} \hat{H} C_\text{z}\T \sum_{t\in I_a} (c(t) \otimes \tau(t)) \mu(t)
= l_\text{a}(t_j) - (c(t_j)\otimes \tau(t_j))\T C_\text{z} \hat{b}, \label{eq:tmpmu}
\end{equation}
for all $t_j \in I_a \subset I_c^k$, and $\mu(t)=0$ for all $t\in I_c^{k}\setminus I_a$, with 
\begin{align}
\hat{H}&:= H^{-1} A_\text{eq}\T (A_\text{eq} H^{-1} A_\text{eq}\T)^{-1} A_\text{eq} H^{-1} - H^{-1},\\
\hat{b}&:= H^{-1} A_\text{eq}\T (A_\text{eq} H^{-1} A_\text{eq}\T)^{-1} b_\text{eq}.
\end{align}
The dual variable $\mu(t)$ is therefore obtained by solving \eqref{eq:tmpmu}, and the optimizer $z^*$ is then determined via \eqref{eq:zstar}. At each iteration, a single constraint is either added or removed. Therefore, in order to efficiently find a solution to \eqref{eq:tmpmu}, a $LDL\T$-decomposition of the matrix
\begin{equation}
\{(c(t_j) \otimes \tau(t_j))\T C_\text{z} \hat{H} C_\text{z}\T (c(t_i) \otimes \tau(t_i)) \}_{(t_j, t_i) \in I_a \times I_a} \label{eq:tmp}
\end{equation}
is computed and adapted at each step by performing rank-1 updates.
The matrix $\hat{H}$ and the vector $\hat{b}$ are precomputed. 
For additional details regarding the regularity of \eqref{eq:tmp}, and issues related to cycling and stalling we refer to \cite[Ch.~10]{Fletcher} and \cite[p.~467]{Nocedal}. 

\subsection{Constraint check}
It remains to explain how to efficiently check whether the constraint \eqref{eq:ineq} is fulfilled for a given solution candidate $z$. We assume that the interval $I_c$ has the form $I_c=[0,T_c]$. As it has been explained in Sec.~\ref{Sec:SemiinfiniteConst}, the constraint check over the interval $[0,\infty)$ reduces to the check over a compact interval, provided that the basis functions fulfill Assumptions A1 and A2. 

A straightforward approach would be to exploit the specific structure of the basis functions. For example, if the basis functions consist of exponentially decaying polynomials having a degree of at most 4, determining the stationary points of \eqref{eq:ineq} amounts to solving a quartic equation, which can be done analytically. As a result, it would be enough to check the constraints at these stationary points in order to determine if the constraint is satisfied or not.


We propose a more general approach that is based on local Taylor approximations, and thus valid for arbitrary basis functions compatible with Assumptions A1 and A2. In order to simplify the discussion, we consider the special case of \eqref{eq:ineq}, where $C_\text{z}$ is the identity and $n_c=1$. The resulting algorithm extends naturally to the more general case. According to Taylor's theorem we obtain the following identity
\begin{equation}
\tau(t)\T z=\tau(0)\T z + \dot{\tau}(0)\T z t + \ddot{\tau}(0)\T z \frac{t^2}{2} + \tau^{(3)}(\bar{t})\T z \frac{t^3}{6},
\end{equation}
where $\bar{t}\in [0,t]$. As will become clear in the following, a third-order Taylor expansion represents a good compromise between approximation quality and computational effort. The last term of the previous equation can be bounded by the Cauchy-Schwarz inequality leading to 
\begin{equation}
|\tau^{(3)}(s)\T z| \leq \sup_{\bar{t} \in I_c} |\tau(\bar{t})|  |(M\T)^3 z| =: R(z), \quad \forall s \in [0,t].
\end{equation}
As a consequence, the following upper and lower bounds are obtained
\begin{equation}
b_\text{l}(t) \leq \tau(t)\T z \leq b_\text{u}(t),
\end{equation}
for all $t\in I_c$, with
\begin{align}
b_\text{l}(t)&:=\tau(0)\T z + \dot{\tau}(0)\T z t + \ddot{\tau}(0)\T z \frac{t^2}{2} - R(z)\frac{t^3}{6}, \\
b_\text{u}(t)&:=\tau(0)\T z + \dot{\tau}(0)\T z t + \ddot{\tau}(0)\T z \frac{t^2}{2} + R(z)\frac{t^3}{6}.
\end{align}
The situation is exemplarily depicted in Fig.~\ref{Fig:illustrationBound} (right). Given that $\dot{\tau}(0)\T z \geq 0$ the lower bound attains its maximum at time 
\begin{equation}
t_\text{u}:=\frac{ \ddot{\tau}(0)\T z + \sqrt{ (\ddot{\tau}(0)\T z)^2 +2 R(z) \dot{\tau}(0)\T z}}{R(z)} > 0,
\end{equation}
whereas if $\dot{\tau}(0)\T z < 0$ the upper bound attains its minimum at time 
\begin{equation}
t_\text{l}:=\frac{- \ddot{\tau}(0)\T z + \sqrt{ (\ddot{\tau}(0)\T z)^2 -2 R(z) \dot{\tau}(0)\T z}}{R(z)} > 0.
\end{equation}
Thus, if the lower bound exceeds $l_\text{u}$ or the upper bound drops below $l_\text{l}$, that is, if
\begin{align}
b_\text{l}(t_\text{u})> l_\text{u}, \quad \text{or}\quad b_\text{u}(t_\text{l}) < l_\text{l},  \label{eq:approxCheck}
\end{align}
the constraint \eqref{eq:ineq} is guaranteed to be violated at time $t_\text{u}$, respectively at time $t_\text{l}$.
If this is not the case, we are guaranteed that the constraint \eqref{eq:ineq} is satisfied in the interval $[0,t_\text{s}]$, with
\begin{align}
t_\text{s}:=\min\{ t_1, t_2 \;|\; b_\text{u}(t_1)=l_\text{u}, b_\text{l}(t_2)=l_\text{l} \}.
\end{align}
Thus, finding the value $t_\text{s}$ requires the solution of two cubic equations, which stems from the fact that a third order Taylor approximation was used as a starting point.
By shifting the parameter vector in time by $t_\text{s}$ and repeating the above procedure, the constraint is either found to be satisfied for all $t\in I_c$ or a constraint violation is detected. Shifting the parameter vector by $t_\text{s}$ amounts in multiplying $z$ with the matrix exponential
\begin{equation}
e^{M\T t_\text{s}} z \rightarrow z. \label{eq:shift}
\end{equation}
The procedure is illustrated by the flow chart depicted in Fig.~\ref{Fig:illustrationBound} (left). 

The efficiency of the proposed strategy can be improved via the following observation. A constraint violation occurring close to $t=0$, is often found within the first few iterations. However, if no constraint violation occurs, the whole interval $I_c$ needs to be traversed, which tends to increase computation. The computational effort may be reduced by including additional conservative constraint satisfaction checks. For example, upper and lower bounds can be tightened by a factor $\gamma\in (0,1)$ such that the satisfaction of
\begin{equation}
\gamma l_\text{l} \leq \tau(t_i)\T z \leq \gamma l_\text{u} \label{eq:qc} 
\end{equation}
for certain time instances $t_i$ implies \eqref{eq:ineq} (for all $t\in I_c$). As a result, at each iteration, the above inequality is checked. If it is found to be fulfilled, then constraint satisfaction can be guaranteed.

\begin{figure}
\begin{minipage}[l]{.45\textwidth}
\center
\scalebox{0.8}{
\tikzstyle{decision} = [diamond, draw, 
    text width=5.5em, text badly centered, node distance=3cm, inner sep=0pt]
\tikzstyle{block} = [rectangle, draw,
    text width=12em, text centered, rounded corners, minimum height=4em]
\tikzstyle{smallblock} = [rectangle, draw,  
    text width=6em, text centered, rounded corners, minimum height=4em]
\tikzstyle{line} = [draw, -latex']
\tikzstyle{cloud} = [draw, ellipse,fill=red!20, node distance=3cm, minimum height=2em]
\tikzstyle{emptyblock} = [text width=6em, text centered]
\def\distv{7em}
\def\disth{10em}

\begin{tikzpicture}[node distance = 2cm, auto]
\node [decision] (eval1) {check \eqref{eq:qc}};
\node [emptyblock,above of=eval1,node distance=\distv] (start) {$I_c=[0,T_c]$, $z$};
\node [emptyblock,left of=eval1,node distance=9em] (satis) {constraint \\ satisfaction};
\node [decision,below of=eval1,node distance=10em] (eval2) {check \eqref{eq:approxCheck}};
\node [emptyblock,left of=eval2,node distance=9em] (viol) { constraint \\ violation};
\node [decision,below of=eval2,node distance=10em] (eval3) {$t+\text{t}_s > T_c$};
\node [smallblock,right of=eval3, node distance=9em] (shift) {$e^{M\T t_\text{s}} z \rightarrow z$ \\
$t+t_s \rightarrow t$};
\node [emptyblock,left of=eval3, node distance=9em] (satis2) { constraint \\ satisfaction};
\path [line] (start) -- node {$0 \rightarrow t$} (eval1);
\path [line] (eval1) -- node[above] {true} (satis);
\path [line] (eval1) -- node[left] {compute $t_\text{u}$, $t_\text{l}$} (eval2);
\path [line] (eval2) -- node[left] {compute $t_\text{s}$} (eval3);
\path [line] (eval2) -- node[above] {true} (viol);
\path [line] (eval3) -- node[above] {true} (satis2);
\path [line] (eval3) -- (shift);
\path [line] (shift) |- (eval1);
\end{tikzpicture}
}
\end{minipage}\hfill
\begin{minipage}[r]{.45\textwidth}
\center
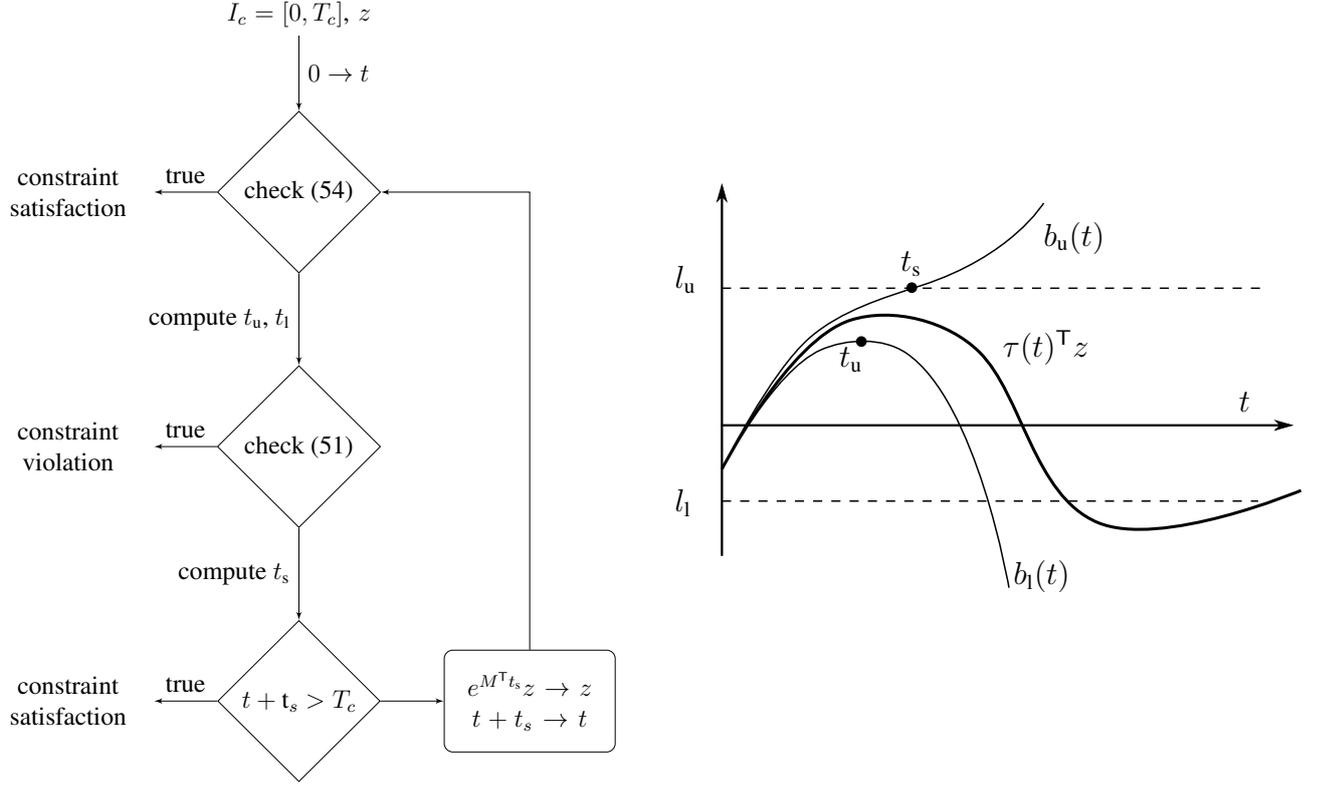
\end{minipage}
\caption{Left: Flow chart that illustrates the proposed constraint satisfaction check. Right: Illustration of the upper and lower bounds $b_\text{u}(t)$ and $b_\text{l}(t)$ obtained from the Taylor expansion of $\tau(t)\T z$. In that case the test \eqref{eq:approxCheck} is indecisive, and constraint satisfaction over the interval $[0,t_\text{s}]$ can be guaranteed.}
\label{Fig:illustrationBound}
\end{figure}

\section{Properties B1-B5}\label{App:PropB}
The sets $\mathcal{C}, \mathcal{C}^s_\text{U}$, and $\mathcal{C}^s_\text{L}$ have the following properties
\begin{itemize}
\item[B1)] the sets $\mathcal{C}^s_\text{U}$, $\mathcal{C}^s_\text{L}$, and $\mathcal{C}$ are closed and convex.
\item[B2)] the sets $\pi^s(\mathcal{C}^s_\text{U})$, $\pi^s(\mathcal{C}^s_\text{L})$ are closed and convex.
\item[B3)] $\mathcal{C}^s_\text{U} \subset \mathcal{C}^{s+1}_\text{U} \subset \mathcal{C}$.
\item[B4)] $\mathcal{C}^s_\text{L} \supset \mathcal{C}^{s+1}_\text{L} \supset \mathcal{C}$.
\item[B5)] $(\pi^s)^* \pi^s( \mathcal{C}^s_\text{U} ) \subset \mathcal{C}^s_\text{U}$, $(\pi^s)^* \pi^s( \mathcal{C}^s_\text{L}) \subset \mathcal{C}^s_\text{L}$.
\end{itemize}
We will sketch the proofs of Properties B1-B5 in the following. 
It follows from the linearity of the stage constraint, and the fact that the terminal constraint $\mathcal{X}$ is convex, that the sets $\mathcal{C}$, $\mathcal{C}^s_\text{U}$, and $\mathcal{C}^s_\text{L}$ are likewise convex.
We will sketch the proof that the set $\mathcal{C}^s_\text{L}$ is closed. The argument can be translated to the set $\mathcal{C}$ by using the following variational formulation
\begin{multline}
\mathcal{C}=\{ (x,u,x_T) \in X \;|\;
\int_{I} \delta p\T (-C_\text{x} x -C_\text{u} u + b) \dt \geq 0 \\
\forall \delta p \in L^2_{n_c}:  \delta p \geq 0, 
\Big|\int_{I} \delta p\T b\dt \Big| < \infty; x_T \in \mathcal{X} \}.\label{eq:mathcalX2}
\end{multline}
We will argue indirectly, i.e. that the complement of $\mathcal{C}^s_\text{L}$ is open. To that extent we choose $(x,u,x_T) \in X \setminus \mathcal{C}^s_\text{L}$. As a result, there exists a test function $\delta \tilde{p}$, with $\delta \tilde{p} \geq 0$, which is spanned by the first $s$ basis functions and is such that
\begin{equation}
\int_{I} \delta \tilde{p}\T(-C_\text{x} x - C_\text{u} u + b)\dt <0.
\end{equation}
For any $\hat{x}\in L^2_n$, $\hat{u}\in L^2_m$, and $\hat{x}_T \in \mathbb{R}^n$, with $||\hat{x}-x|| <\epsilon$, $||\hat{u}-u|| < \epsilon$, and $|\hat{x}_T - x_T| < \epsilon$, it follows that
\begin{align}
\int_{I} \delta \tilde{p}\T (-C_\text{x} \hat{x} - C_\text{u} \hat{u} + b) \dt = &\int_{I} \delta \tilde{p}\T (-C_\text{x} x - C_\text{u} u + b) \dt +\int_{I} \delta{\tilde{p}}\T (C_\text{x} (x-\hat{x}) + C_\text{u} (u-\hat{u})) \dt,
\end{align}
where the last integral can be bounded by (using the Cauchy-Schwarz inequality)
\begin{equation}
||\delta \tilde{p}||(|C_\text{x}| ||x-\hat{x}|| + |C_\text{u}| ||u-\hat{u}||) < ||\delta \tilde{p}|| (|C_\text{x}| + |C_\text{u}|) \epsilon.
\end{equation}
We can infer from $\mathcal{X}$ being closed, that there exists an open ball centered at $x_T$, which does not intersect $\mathcal{X}$. As a result, by choosing $\epsilon$ small enough, it can be concluded that
\begin{equation}
\int_{I} \delta \tilde{p}\T(-C_\text{x} \hat{x} - C_\text{u} \hat{u} + b)\dt<0, \quad \hat{x}_T \not\in \mathcal{X},
\end{equation}
and therefore $(\hat{x},\hat{u},x_T)\in X \setminus \mathcal{C}^s_\text{L}$, for all $\hat{x} \in L^2_n$, $\hat{u} \in L^2_m$, $\hat{x}_T \in \mathbb{R}^n$ with $||x-\hat{x}|| < \epsilon$, $||u-\hat{u}|| < \epsilon$, and $|\hat{x}_T-x_T|<\epsilon$. Hence, the complement of $\mathcal{C}^s_\text{L}$ is open, and therefore $\mathcal{C}^s_\text{L}$ is closed.

Note that the set $\mathcal{C}^s_\text{U}$ is given by the intersection of the set $\mathcal{C}$ with the linear subspace $X^s$ spanned by the first $s$ basis functions. Both of these sets are closed\footnote{$X^s$ is finite-dimensional, thus complete, and hence also closed.} implying that $\mathcal{C}^s_\text{U}$ is closed as well.

The projection $\pi^s$ is linear, which asserts the convexity of the sets $\pi^s(\mathcal{C}^s_\text{U})$ and $\pi^s(\mathcal{C}^s_\text{L})$. Moreover, it is surjective, and hence, by the open mapping theorem, it follows directly from $\mathcal{C}^s_\text{L}$ and $\mathcal{C}^s_\text{U}$ being closed that $\pi^s(\mathcal{C}^s_\text{L})$ and $\pi^s(\mathcal{C}^s_\text{U})$ are closed as well.

The inclusion $\mathcal{C}^s_\text{U} \subset \mathcal{C}^{s+1}_\text{U}$ follows directly from $\mathcal{C}^s_\text{U}=\mathcal{C} \cap X^s$ and the inclusion $X^s \subset X^{s+1}$. In other words, given $(x,u,x_T) \in \mathcal{C}^s_\text{U}$, the parameter vectors $\eta_x$ and $\eta_u$, corresponding to the state and input trajectories $x$ and $u$, can be extended with zeros resulting in trajectories $\hat{x}$, $\hat{u}$ spanned by $s+1$ basis functions. But $\hat{x}=x$ and $\hat{u}=u$ and therefore $(x,u,x_T)\in \mathcal{C}^{s+1}_\text{U}$. The inclusion $\mathcal{C}^s_\text{L} \supset \mathcal{C}^{s+1}_\text{L}$ follows from the fact that the dimension of the subspace to which the test functions $\delta \tilde{p}$ are constrained increases with $s$.

The claim $(\pi^s)^* \pi^s(\mathcal{C}^s_\text{U}) \subset \mathcal{C}^s_\text{U}$ follows from the fact that $(\pi^s)^* \pi^s$ is a projection from $X$ onto $X^s \subset X$ and that $\mathcal{C}^s_\text{U} \subset X^s$. The claim that $(\pi^s)^* \pi^s(\mathcal{C}^s_\text{L}) \subset \mathcal{C}^s_\text{L}$ follows from the linearity of the stage constraints. More precisely, it follows by noting that for any $\delta \tilde{p}=(I_{n_c} \otimes \tau)\T \delta \eta_p$ and any $x\in L^2_n$,
\begin{align}
\int_{I} \delta \tilde{p}\T C_\text{x} x\dt &= \delta \eta_p\T  C_\text{x} \pi^{ns}(x) \nonumber \\
&= \int_{I} \delta \tilde{p}\T C_\text{x} ~(\pi^{ns})^* \pi^{ns}(x) \dt
\end{align}
holds.

\section{Properties C1-C5}\label{App:PropC}
The sets $\mathcal{D}$, $\mathcal{D}^s_\text{U}$, and $\mathcal{D}^s_\text{L}$ have the following properties
\begin{itemize}
\item[C1)] the sets $\mathcal{D}^s_\text{U}$, $\mathcal{D}^s_\text{L}$, and $\mathcal{D}$ are closed and convex.
\item[C2)] the sets $\pi^s(\mathcal{D}^s_\text{U})$ and $\pi^s(\mathcal{D}^s_\text{L})$ are closed and convex.
\item[C3)] $\mathcal{D}^s_\text{U} \subset \mathcal{D}^{s+1}_\text{U} \subset \mathcal{D}$.
\item[C4)] $\mathcal{D}^s_\text{L} \supset \mathcal{D}^{s+1}_\text{L} \supset \mathcal{D}$.
\item[C5)] $(\pi^s)^* \pi^s (\mathcal{D}^s_\text{U}) \subset \mathcal{D}^s_\text{U}$, $(\pi^s)^* \pi^s(\mathcal{D}^s_\text{L}) \subset \mathcal{D}^s_\text{L}$.
\end{itemize}

We will sketch the proof of Properties C1-C5 below. 

The convexity of the sets $\mathcal{D}^s_\text{U}$, $\mathcal{D}^s_\text{L}$, and $\mathcal{D}$ follows directly from the linearity of the dynamics.

The fact that the set $\mathcal{D}^s_\text{L}$ is closed can be seen by a similar argument used for showing closedness of $\tilde{\mathcal{C}}^s$ in Sec.~\ref{App:PropB}, that is, showing that $X\setminus \mathcal{D}^s_\text{L}$ is open. The set $\mathcal{D}$ can be characterized using the variational equality, 
\begin{align}
-\int_{I} \delta \dot{p}\T x \dt - \int_{I}& \delta p\T(A x + B u) \dt 
- \delta p(0)\T x_0 +\delta p(T)\T x_T=0, \quad \forall \delta p \in H_n, \label{eq:weak}
\end{align}
where $H_n$ denotes the set of functions in $L^2_n$ having a weak derivative in $L^2_n$. Thus, again a similar argument can be applied to show that $\mathcal{D}$ is closed. It follows that $\mathcal{D}^s_\text{U}$ is closed, since $\mathcal{D}^s_\text{U}$ is defined as the intersection of $\mathcal{D}$ with the closed set $X^s$.

The linearity and surjectivity of the map $\pi^s$ implies that $\pi^s(\mathcal{D}^s_\text{U})$ and $\pi^{s}(\mathcal{D}^s_\text{L})$ are indeed closed (by the open mapping theorem) and convex. 

The inclusion $\mathcal{D}^s_\text{U} \subset \mathcal{D}^{s+1}_\text{U} \subset \mathcal{D}$ for all $s$ follows directly from the fact that $X^s \subset X^{s+1} \subset X$ and $\mathcal{D}^s_\text{U}=\mathcal{D} \cap X^s$. The inclusion $\mathcal{D}^{s+1}_\text{L} \subset \mathcal{D}^s_\text{L}$ for all $s$ can be seen by noting that the variational equality in \eqref{eq:deftD} has to hold for variations spanned by more and more basis functions as $s$ increases. The claim that $\mathcal{D}$ is contained in $\mathcal{D}^s_\text{L}$ follows from the equivalence of \eqref{eq:weak} with the formulation in \eqref{eq:defD}. 
The properties $(\pi^s)^* \pi^s(\mathcal{D}^s_\text{U}) \subset \mathcal{D}^s_\text{U}$ and $(\pi^s)^* \pi^s(\mathcal{D}^s_\text{L}) \subset \mathcal{D}^s_\text{L}$ can be shown using the same arguments as in \ref{App:PropB}, where the latter relies on the linearity of the dynamics.
\section{Convergence results}\label{App:SetConv}
We prove the following result for the case where $I$ has infinite measure.

\begin{lemma}
Given that the basis functions form an algebra and that the basis functions are dense in the set of smooth functions with compact support in $I$, it holds that
\begin{equation*}
\lim_{s \rightarrow \infty} \mathcal{C}^s_\text{U} = \lim_{s \rightarrow \infty} \mathcal{C}^s_\text{L}.
\end{equation*}
\end{lemma}
\begin{proof}
The idea of the proof is the following: We claim that $\lim_{s\rightarrow \infty} \mathcal{C}^s_\text{U} \supset \lim_{s \rightarrow \infty} \mathcal{C}^s_\text{L}$. We assume that the claim is incorrect and show that this leads to a contradiction. Thus, we choose $(x,u,x_T) \in \lim_{s\rightarrow \infty} \mathcal{C}^s_\text{L}$, such that there exists an open set $U$ (bounded) and a $k\in \{1,2,\dots,n_c\}$, for which 
\begin{equation}
\int_{I} \delta v ( - C_{\text{x}k} x -C_{\text{u}k} u + b_k ) \dt < 0 \label{eq:tmpbegin}
\end{equation}
holds for all smooth test functions $\delta v: I \rightarrow \mathbb{R}$, $\delta v \geq 0$, with support in $U$, and $\delta v(t_0)>0$ for some $t_0 \in U$. Due to the smoothness of the test functions, $\delta v(t_0)>0$ readily implies that there is an open neighborhood of $t_0$, denoted by $\mathcal{N}(t_0)$, such that $\delta v(t)>0$, $\forall t\in \mathcal{N}(t_0)$.
The above integral exists, since $\delta v$ is bounded, has compact support and $x\in L^2_n$, $u\in L^2_m$. We fix $t_0 \in U$ and pick one of these variations that is positive, strictly positive at time $t_0$, and has support in $U$, which we name $\delta p$.
Due to the fact that the basis functions are dense in the set of smooth functions with compact support, there exists a sequence $\sqrt{\delta \tilde{p}_i}$ that converges uniformly to $\sqrt{\delta p}$. Due to the fact that the basis functions form an algebra, $\delta \tilde{p}_i$ lies likewise in the span of the basis functions. Moreover, for a given $\epsilon > 0$ (small enough) there exits an integer $N>0$ such that
\begin{equation*}
||\delta \tilde{p}_i - \delta p||_\infty < C_1 \epsilon
\end{equation*}
holds for all integers $i> N$, where $C_1>0$ is constant. 

We claim that there is an integer $p>0$ such that the basis function $\tau_p: I \rightarrow \mathbb{R}$ is nonzero for $t_0$. This can be shown by a contradiction argument: If the claim was not true, then all basis functions $\tau^p=(\tau_1,\tau_2,\dots,\tau_p)$ would be zero at time $t_0$. The basis functions fulfill Assumption A2, and hence the first order differential equation $\dot{\tau}^p = M_p \tau^p$, which is guaranteed to have unique solutions. From $\tau^p(t_0)=0$, $\dot{\tau}^p(t_0)=0$ we can infer that $\tau^p(t)=0$ for all $t\in (0,\infty)$ is the (unique) solution  to $\dot{\tau}^p=M_p \tau^p$, contradicting Assumption A1. 

Thus, we can choose the basis function $\tau_p$ that is nonzero for $t_0$ and hence also nonzero in a neighborhood around $t_0$, due to the smoothness of the basis functions. The same applies to the function $\tau_p^2$, which is likewise contained in the set of basis functions (the basis functions form an algebra that is closed under multiplication). Moreover, due to the fact that the basis functions are orthonormal, it follows that
\begin{equation}
\int_{I} \tau_p^2 \dt = 1. \label{eq:normalizationp}
\end{equation} 
In the same way, the function $\delta \tilde{p}_i(t) \tau_p^2(t)$ is also contained in the set of basis functions, is non-negative for all $t\in I$, and is bounded and integrable for all integers $i>0$ ($|\tau^p|$ is bounded, see \ref{App:BasisFunProp}). Thus, it follows that the integral
\begin{equation}
\int_{I} \tau_p^2 \delta \tilde{p}_i (-C_{\text{x}k} x - C_{\text{u}k} u + b_k) \dt
\end{equation}
exists.
By assumption $(x,u,x_T)\in \lim_{s\rightarrow \infty} \mathcal{C}^s_\text{L}$, and therefore
\begin{align}
0&\leq \int_{I} \tau_p^2 \delta \tilde{p}_i(-C_{\text{x}k} x - C_{\text{u}k} u + b_k) \dt\\
&=\int_{I} \tau_p^2 \delta p(-C_{\text{x}k} x - C_{\text{u}k} u + b_k) \dt \label{eq:tmpneg2}\\
&+\int_{I} \tau_p^2 (\delta \tilde{p}_i - \delta p) (- C_{\text{x}k} x - C_{\text{u}k} u + b_k) \dt \label{eq:tmpsmall2},
\end{align}
where the last term can be bounded by
\begin{equation}
\epsilon C_1 \int_{I} \tau_p^2 |C_{\text{x}k} x + C_{\text{u}k} u - b_k |\dt.
\end{equation}
The above integral is bounded due to the fact that $\tau_p^2$ is bounded and integrable, $x\in L_n^2$, $u\in L_m^2$, and that $b_k$ is bounded, and therefore \eqref{eq:tmpsmall2} can be made arbitrarily small by sufficiently increasing $i$. However, this leads to a contradiction, since \eqref{eq:tmpneg2} is strictly negative, according to \eqref{eq:tmpbegin}. This proves that $\lim_{s\rightarrow \infty} \mathcal{C}^s_\text{U} \supset \lim_{s\rightarrow \infty}\mathcal{C}^s_\text{L}$. The desired result is then established due to the fact that $\mathcal{C}^s_\text{U} \subset \mathcal{C}^s_\text{L}$ holds for all integers $s>0$.\hfill \qed
\end{proof}
\section{Recursive feasibility and closed-loop stability}
\label{App:Stab}
\begin{proposition}
Provided that the optimization \eqref{eq:mpc1} is feasible at time $t=0$, it remains feasible for all times $t>0$, and the resulting closed-loop system is guaranteed to be asymptotically stable.
\end{proposition}
\begin{proof}
The proof is taken from \cite{muehlebachParametrized} and included for completeness. The following notation is introduced: The closed-loop state and input trajectories are denoted by $x(t)$ and $u(t)$. The predicted trajectories are referred to as $\tilde{x}(t_\text{p}|t)$, $\tilde{u}(t_\text{p}|t)$, where $t_\text{p}>0$ denotes the prediction horizon. For $t_\text{p}=0$, the prediction matches the true trajectory, that is $\tilde{x}(0|t)=x(t)$, $\tilde{u}(0|t)=u(t)$ for all $t\in [0,\infty)$. The predictions $\tilde{x}(t_\text{p}|t)$, $\tilde{u}(t_\text{p}|t)$ are obtained by solving \eqref{eq:mpc1} subject to the initial condition $x_0=x(t)$, which yields the parameters $\eta_x$ and $\eta_u$ defining $\tilde{x}(t_\text{p}|t)$ and $\tilde{u}(t_\text{p}|t)$ by
\begin{equation}
\tilde{x}(t_\text{p}|t)=(I_n \otimes \tau(t_\text{p}))\T \eta_x, \quad \tilde{u}(t_\text{p}|t)=(I_m \otimes \tau(t_\text{p}))\T \eta_u.
\end{equation}
In order to highlight the dependence on the initial condition, the resulting optimal cost of \eqref{eq:mpc1} is denoted by $J^{\text{MPC}}(x(t))$.

By assumption, \eqref{eq:mpc1} is feasible at time $t=0$. The resulting trajectories $\tilde{x}(t_\text{p}|0)$, $\tilde{u}(t_\text{p}|0)$ fulfill the equations of motion, the initial condition $\tilde{x}(0|0)=x(0)$, and the constraints and hence, the system evolves according to $x(t)=\tilde{x}(t|0)$, $u(t)=\tilde{u}(t|0)$, $\forall t\in [0,T_\text{d})$. Due to the time-shift property of the basis functions implied by Assumption A2, the feasible candidates 
\begin{align}
\begin{split}
\tilde{x}(t_\text{p}+T_\text{d}|0)=(I_n \otimes \tau(t_\text{p})\T) (I_n \otimes \exp(M_s T_\text{d})\T)\eta_x,\\
\tilde{u}(t_\text{p}+T_\text{d}|0)=(I_m \otimes \tau(t_\text{p})\T) (I_m \otimes \exp(M_s T_\text{d})\T)\eta_u
\end{split}
\label{eq:abovetmp}
\end{align}
for the optimization at time $T_\text{d}$ can be constructed from the optimizer $\eta_x$ and $\eta_u$ at time $0$. As a result, recursive feasibility of \eqref{eq:mpc1} follows by induction.

We will show that the optimal cost $J^{\text{MPC}}$ acts as a Lyapunov function. The function $J^{\text{MPC}}$ is a valid Lyapunov candidate since $J^{\text{MPC}}(x)>0$ for all $x \neq 0$ and $J^{\text{MPC}}(x)=0$ if and only if $x=0$. Due to the fact that the shifted trajectories $\tilde{x}(t_\text{p}+T_\text{d}|0)$ and $\tilde{u}(t_\text{p}+T_\text{d}|0)$ (as defined in \eqref{eq:abovetmp}) are feasible for the optimization at time $T_\text{d}$, the following upper bound on $J^{\text{MPC}}(x(T_\text{d}))$ can be established
\begin{equation}
J^{\text{MPC}}(x(T_\text{d})) \leq \int_{T_\text{d}}^{\infty} \!\!\!\! \tilde{x}(t_\text{p}|0)\T Q \tilde{x}(t_\text{p}|0) + \tilde{u}(t_\text{p}|0)\T R \tilde{u}(t_\text{p}|0) \dt_\text{p}.
\end{equation}
The right-hand side can be rewritten as
\begin{equation}
J^{\text{MPC}}(x(0)) - \int_{0}^{T_\text{d}} \!\!\!\! \tilde{x}(t_\text{p}|0)\T Q \tilde{x}(t_\text{p}|0) + \tilde{u}(t_\text{p}|0)\T R \tilde{u}(t_\text{p}|0) \dt_\text{p},
\end{equation}
resulting in
\begin{align}
J^{\text{MPC}}&(x(T_\text{d})) - J^{\text{MPC}}(x(0)) \leq \label{eq:rhsreftmp} \\
&-\int_{0}^{T_\text{d}} \tilde{x}(t_\text{p}|0)\T Q \tilde{x}(t_\text{p}|0) + \tilde{u}(t_\text{p}|0)\T R \tilde{u}(t_\text{p}|0) \dt_\text{p}. \nonumber
\end{align}
The right-hand side of \eqref{eq:rhsreftmp} is guaranteed to be strictly negative, except for $x(0)=0$, and thus, by induction, $J^{\text{MPC}}(x(kT_\text{d}))$ is strictly decreasing, which concludes the proof. \hfill \qed
\end{proof}
\section{Reduction of the semi-infinite constraint}\label{App:Decay}
The following section discusses the reduction of the semi-infinite constraint 
\begin{equation}
a_\text{l} \leq \tau(t)\T \eta \leq a_\text{u}
\end{equation}
over the interval $t \in [0,\infty)$, with $a_\text{l},a_\text{u}\in \mathbb{R}$, $a_\text{l}<0$, $a_\text{u}>0$ to a compact interval. Thereby we consider the symmetric case, where $|a_\text{l}|=|a_\text{u}|$ first, before discussing the asymmetric case $|a_\text{l}|\neq |a_\text{u}|$. It is shown that the length of this compact time interval depends only on the properties of the basis functions and on the ratio between $|a_\text{u}|$ and $|a_\text{l}|$. Both $a_\text{l}$ and $a_\text{u}$ are assumed to be finite. 
\subsection{The symmetric case}
\begin{proposition}\label{Prop:Decay}
Provided that the basis functions fulfill Assumptions A1 and A2 for all $t \in [0,\infty)$ there exists a positive real number $T_\text{c}$ such that
\begin{equation}
\sup_{t\in [0,\infty)} |\tau(t)\T \eta| = \max_{t\in [0,T_\text{c}]} |\tau(t)\T \eta|
\end{equation}
holds for all parameter vectors $\eta \in \mathbb{R}^s$.
\end{proposition}
\begin{proof}
Without loss of generality we restrict the parameter vectors to have unit magnitude, that is, $|\eta|=1$.\footnote{The claim holds trivially for $\eta=0$; in case $\eta\neq 0$ we can always normalize $\eta$.}

We prove the claim in 4 steps. We first derive an exponentially decaying upper bound on the Euclidean norm of the basis function vector $\tau$. We then use this bound to argue that the basis functions are linearly independent over the interval $[0,T_\text{i}]$ (the scalar $T_\text{i}$ will be determined). The third step consists of constructing a lower bound on 
\begin{equation}
\max_{t\in [0,T_\text{i}]} |\tau(t)\T \eta|
\end{equation}
that holds for all parameter vectors $\eta$ with $|\eta|=1$. Linear independence of the basis functions on $[0,T_\text{i}]$ will be used to argue that this lower bound is strictly positive. In the last step, we show that if $t$ is sufficiently large, $|\tau(t)\T \eta|$ will be below this lower bound, which concludes the proof.

Step 1): The fact that $M_s$ is asymptotically stable implies that there exists a quadratic Lyapunov function that decays exponentially. This provides a means to establish the following bound
\begin{equation}
|\tau(t)|^2 \leq C_2 e^{-c_2 t},\quad \forall t\in [0,\infty),  \label{eq:BoundTau}
\end{equation}
where $C_2>0$, $c_2>0$ are constant. 

Step 2): We use the Gram matrix to argue that the basis functions are linearly independent over the interval $[0,T_\text{i}]$. According to \cite[p.~2, Thm.~3]{SansoneOrtho} it holds that the basis functions are linearly independent in the set $[0,T_\text{i}]$ if and only if the matrix 
\begin{equation}
\int_{0}^{T_\text{i}} \tau \tau\T \dt\label{eq:Gram}
\end{equation}
has full rank. This is the case if the bilinear form
\begin{equation}
v\T  \int_{0}^{T_\text{i}} \tau \tau\T \dt~v
\end{equation}
is strictly positive for all $v \in \mathbb{R}^s$ with $|v|=1$. Combining the fact that the basis functions are orthonormal with the Cauchy-Schwarz inequality, leads to the following lower bound of the above bilinear form,
\begin{equation}
1- \int_{T_\text{i}}^{\infty} (v\T \tau)^2 \dt \geq 1- \int_{T_\text{i}}^{\infty} |\tau|^2 \dt.
\end{equation}
Using the upper bound \eqref{eq:BoundTau} we therefore obtain
\begin{equation}
v\T  \int_{0}^{T_\text{i}} \tau \tau\T \dt v\geq 1- \int_{T_\text{i}}^{\infty} C_2 e^{-c_2 t} \dt = 1-\frac{C_2}{c_2} e^{-c_2 T_\text{i}},
\end{equation}
for all $v \in \mathbb{R}^s$ with $|v|=1$.
Thus, we fix $T_\text{i}>0$, such that 
\begin{equation}\label{eq:posDef1}
1> \frac{C_2}{c_2} e^{-c_2 T_\text{i}},
\end{equation}
implying that the matrix \eqref{eq:Gram} is positive definite and has therefore full rank. Consequently, the basis functions are guaranteed to be linearly independent on the interval $[0,T_\text{i}]$.

Step 3): We claim that 
\begin{equation}
c^*:=\inf_{|\eta|=1} \max_{t \in [0,T_\text{i}]} |\tau(t)\T \eta| \label{eq:infeta}
\end{equation}
is well-defined and strictly positive. To that extent we first prove that the function $g: \mathbb{R}^s \rightarrow [0,\infty)$, 
\begin{equation}
g(\eta):=\max_{t\in [0,T_\text{i}]} |\tau(t)\T \eta|
\end{equation}
is continuous (in fact Lipschitz-continuous). Therefore we consider two parameter vectors $\eta_1$ and $\eta_2$ with $g(\eta_1) \leq g(\eta_2)$ (without loss of generality).
From the fact that $g(\eta_2)$ can be rewritten as $g(\eta_1 + (\eta_2-\eta_1))$ and by invoking the triangle inequality it can be inferred that
\begin{equation}
g(\eta_1 + (\eta_2-\eta_1)) \leq \max_{t \in [0,T_\text{i}]} |\tau(t)\T \eta_1| + |\tau(t)\T (\eta_2-\eta_1)|.
\end{equation}
By noting that the maximum of the above sum is smaller than the sum of the summand's maxima it can be concluded that
\begin{equation}
g(\eta_2) \leq g(\eta_1) + \max_{t \in [0,T_\text{i}]} |\tau(t)\T (\eta_2-\eta_1)|.
\end{equation}
Moreover, by combining the Cauchy-Schwarz inequality and the bound \eqref{eq:BoundTau} we obtain
\begin{equation}
|g(\eta_2)-g(\eta_1)| \leq \sqrt{C_2} |\eta_2-\eta_1|,
\end{equation}
showing that the function $g$ is indeed (Lipschitz) continuous. As a result, the Bolzano-Weierstrass theorem asserts that the infimum in \eqref{eq:infeta}, is attained and well-defined. 
It remains to argue that $c^* > 0$. For the sake of contradiction, we assume $c^*=0$. This implies the existence of the minimizer $\eta^*$, with $|\eta^*|=1$, which fulfills
\begin{equation}
\max_{t\in [0,T_\text{i}]} |\tau(t)\T \eta^*|=0.
\end{equation}
As a result, it follows that $\tau(t)\T \eta^*$ is zero for all $t\in [0,T_\text{i}]$, which contradicts the fact that the basis functions $\tau$ are linearly independent on $[0,T_\text{i}]$.

Step 4): From the upper bound \eqref{eq:BoundTau} and the Cauchy-Schwarz inequality it follows that
\begin{equation}
|\tau(t)\T \eta| \leq \sqrt{C_2} e^{-\frac{c_2 t}{2}}, \quad \forall t\in [0,\infty),
\end{equation}
and for all $\eta\in \mathbb{R}^s$ with $|\eta|=1$. Clearly, $c^*\leq \sqrt{C_2}$ and therefore we can choose the time $T_\text{c}$ such that 
\begin{equation}
c^*=\sqrt{C_2} e^{-\frac{c_2 T_\text{c}}{2}}, 
\end{equation}
implying
\begin{equation}
\sup_{t\in (T_\text{c},\infty)} |\tau(t)\T \eta| < c^* \leq \max_{t\in [0,T_\text{i}]} |\tau(t)\T \eta|,
\end{equation}
for all $\eta \in \mathbb{R}^s$ with $|\eta|=1$. This proves the claim.
\end{proof}

\subsection{The asymmetric case}
\begin{proposition}\label{Prop:DecayAsym}
Provided that the basis functions fulfill Assumptions A1 and A2 for all $t\in [0,\infty)$ there exists a positive real number $T_c$ such that 
\begin{equation}
a_\text{l} \leq \tau(t)\T \eta \leq a_\text{u}, \quad \forall t\in [0,T_c]
\end{equation}
implies 
\begin{equation}
a_\text{l} \leq \tau(t)\T \eta \leq a_\text{u}, \quad \forall t\in [0,\infty), 
\end{equation}
for any $\eta \in \mathbb{R}^s$, where $a_\text{l}, a_\text{u} \in \mathbb{R}$, $a_\text{l}<0, a_\text{u}>0$.
\end{proposition}
\begin{proof}
We define $\tilde{f}:= \tau\T \eta$ and establish upper and lower bounds on $\tilde{f}$. Without loss of generality we assume $\eta \neq 0$. Combining Assumption A2 with the Caley-Hamilton theorem leads to
\begin{align}
\tilde{f}^{(s)}(t) + a_1 \tilde{f}^{(s-1)}(t) + \dots + a_s \tilde{f}(t)&=
\eta\T \tau^{(s)}(t) + a_1 \eta\T \tau^{(s-1)}(t) + \dots + a_s \eta\T \tau(t)\\
&=\eta\T \left( M^s + a_1 M^{s-1} + \dots + a_s M \right) \tau(t)=0,
\end{align}
where $a_1, a_2, \dots a_s$ are the coefficients of the characteristic polynomial of the matrix $M$. Thus, the trajectory $\tilde{f}$ and its time derivatives fulfill the following set of differential equations
\begin{equation}
\dot{f}=\underbrace{\left( \begin{array}{ccccc} 0 & 1 & 0 & \dots & 0\\
									0 & 0 & 1 & \dots & 0\\
									\vdots  &  \vdots &   & \ddots & \vdots\\
									-a_s & -a_{s-1} & -a_{s-2} & \dots & -a_1\end{array} \right)}_{:=\hat{M}} f, \quad f(0)=\underbrace{\left(\begin{array}{c} \tau(0)\T \\
\tau(0)\T M\T \\
\vdots \\
\tau(0)\T (M^{(s-1)})\T \end{array}\right)}_{:=\hat{H}} \eta,
\end{equation}
where $f:=(\tilde{f}, \tilde{f}^{(1)}, \dots, \tilde{f}^{(s-1)})$. The matrix $\hat{M}$ is Hurwitz and therefore, due to the Lyapunov theorem, there exists a symmetric matrix $P>0$, $P \in \mathbb{R}^{s \times s}$ that satisfies
\begin{equation}
P \hat{M} + \hat{M}\T P + Q = 0 \label{eq:AppLyap}
\end{equation}
for any symmetric matrix $Q>0$, $Q \in \mathbb{R}^{s \times s}$. We fix the positive definite matrix $Q$ and consider the quadratic Lyapunov function $V(t)=f(t)\T P f(t)$, where $P$ satisfies \eqref{eq:AppLyap}. The time derivative of $V$ can be upper bounded by
\begin{align}
\dot{V}= - f\T Q f \leq -\lambda_Q |f|^2 &\leq -\frac{\lambda_Q}{\lambda^P} f\T P f\\
& \leq - \frac{\lambda_Q}{\lambda^P} V,
\end{align}
where the minimum eigenvalue of $Q$ is denoted by $\lambda_Q$ and the maximum eigenvalue of $P$ is denoted by $\lambda^P$. As a result, this yields the upper bound
\begin{equation}
V(t) \leq V(0) e^{-\frac{\lambda_Q}{\lambda^P} t} \leq \lambda^P \lambda^{\hat{H}\T \hat{H}} |\eta|^2 e^{-\frac{\lambda_Q}{\lambda^P} t},
\end{equation}
where $\lambda^{\hat{H}\T \hat{H}}$ denotes the maximum eigenvalue of the matrix $\hat{H}\T \hat{H}$.
According to the proof of Prop.~\ref{Prop:Decay}, we may choose $T_\text{i}$ such that 
\begin{equation}
1 > \frac{C_2}{c_2} e^{-c_2 T_\text{i}},
\end{equation}
implying that the basis functions are linearly independent on the interval $[0,T_\text{i}]$ (see \eqref{eq:posDef1} and Prop.~\ref{Prop:Decay} for the definition of the constants $c_2$ and $C_2$). Linear independence can be used to establish the following lower bound, c.f. Prop.~\ref{Prop:Decay}:
\begin{equation}
\int_{0}^{T_\text{i}} |\tilde{f}|^2 \dt = \eta\T \int_{0}^{T_\text{i}} \tau \tau\T \dt~ \eta \geq c_3 |\eta|^2,
\end{equation}
where the constant $c_3>0$ denotes the minimum eigenvalue of the matrix 
\begin{equation}
\int_{0}^{T_\text{i}} \tau \tau\T \dt.
\end{equation}
Without loss of generality it is assumed that $|a_\text{l}| \leq |a_\text{u}|$. Choosing the real number $T_\text{c}>T_\text{i}$ implies that the constraint is imposed over the interval $[0,T_\text{i}]$ and therefore
\begin{equation}
\int_{0}^{T_\text{i}} |\tilde{f}|^2 \dt \leq T_\text{i} |a_\text{u}|^2.
\end{equation}
Combined with the above upper bound on $|\eta|^2$ it follows that
\begin{equation}
|\eta|^2 \leq \frac{T_\text{i} |a_\text{u}|^2}{c_3}.
\end{equation}
This can be used to upper bound the squared magnitude of $\tilde{f}(t)$, that is,
\begin{align}
|\tilde{f}(t)|^2 \leq \frac{1}{\lambda_P} V(t) \leq \frac{\lambda^P \lambda^{\hat{H}\T \hat{H}} |a_\text{u}|^2 T_\text{i}}{\lambda_P c_3} e^{-\frac{\lambda_Q}{\lambda^P} t},
\end{align}
where $\lambda_P$ denotes the minimum eigenvalue of the matrix $P$.
As a result, we may choose $T_\text{c}$ such that the above upper bound is below $|a_\text{l}|^2$. This may be achieved by choosing $T_\text{c} > \max\{ \hat{t}, T_\text{i}\}$, where 
\begin{equation}
\hat{t}:= - \frac{\lambda^P}{\lambda_Q} ~\left(2~\text{ln}\left( \frac{|a_\text{u}|}{|a_\text{l}|}\right) ~+ \text{ln}\left(\frac{c_3 \lambda_P}{T_\text{i} \lambda^P \lambda^{\hat{H}\T \hat{H}}}\right) \right)
\end{equation}
and $\text{ln}$ refers to the natural logarithm.
\end{proof}

\section{Additional properties}\label{App:BasisFunProp}
In the following we will discuss some of the properties of the basis functions that fulfill Assumptions A1 and A2.

\subsection{Conditions imposed on $M_s$}
The fact that the basis functions are orthonormal on the interval $I=(0,T)$ implies
\begin{align}
\tau(T)\tau(T)\T - \tau(0) \tau(0)\T &= \int_{I} \frac{\diff}{\dt} (\tau \tau\T) \dt\\
&= M_s + M_s\T.
\end{align}
In case the interval $I$ has infinite measure, that is, $I=(0,\infty)$, the above formula reduces naturally to 
\begin{equation}
-\tau(0) \tau(0)\T = M_s + M_s\T. \label{eq:condMsinf}
\end{equation}

\subsection{Bounds on the Euclidean norm}
In case the interval $I$ has infinite measure, that is, $I=(0,\infty)$, it holds that 
\begin{equation}
|\tau(t)|\leq |\tau(0)|, \quad \forall t\in [0,\infty).
\end{equation}
This results from
\begin{align}
\frac{\diff}{\dt}( \tau(t)\T \tau(t) )&= \tau(t)\T (M_s + M_s\T) \tau(t)\\
&=-|\tau(0)\T \tau(t)|^2 \leq 0,
\end{align}
for all $t\in [0,\infty)$, which follows from \eqref{eq:condMsinf}.


\subsection{Finite number of minima and maxima}
In the following we will argue that a function that is not everywhere zero and spanned by basis functions fulfilling Assumptions A1 and A2, has a finite number of minima and maxima. The function will be denoted by $\tilde{f}:=\tau\T \eta$, $\eta \in \mathbb{R}^s, \eta \neq 0$, where the function and the basis functions are defined over the interval $t\in I=(0,T)$. If $I$ has infinite measure, it follows from Sec.~\ref{App:Decay} that $\tilde{f}$ takes is maxima and minima within a compact interval, and hence, without loss of generality, we assume in the following that $I$ has finite measure. Moreover, due to the Caley-Hamilton theorem it holds that
\begin{equation}
\tilde{f}^{(s)}(t) + a_1 \tilde{f}^{(s-1)}(t) + \dots + a_s \tilde{f}(t) = 0,
\end{equation}
for all $t\in I$, where the $a_k$, $k=1,2,\dots,s$ are the coefficients of the characteristic polynomial of the matrix $M_s$. According to \cite{FunCurves}, the time derivative of $\tilde{f}$ has at most $s-1$ zeros on any subinterval of length $l$, where 
\begin{equation}
\sum_{k=1}^{s} \frac{a_k l^k}{k!} < 1.
\end{equation}
This proves readily that $\tilde{f}$ has a finite number of minima and maxima in the interval $I$.

\bibliography{literature}
\bibliographystyle{abbrv}

\end{document}